\documentclass[journal]{IEEEtran}

\usepackage{graphicx}
\usepackage{amsfonts}
\usepackage{amsmath}
\usepackage{amssymb}
\usepackage{amsthm}
\usepackage{subfig}
\usepackage[font=small]{caption}
\usepackage{stfloats}
\usepackage{amsthm}
\usepackage{relsize}
\usepackage{rotating}
\usepackage{pgfplots}
\usepackage{siunitx}
\usepackage{cite}

\usepackage{hyperref}
\usepackage{fancyhdr}

\newtheorem{theorem}{Theorem}

\begin{document}

%

\title{Denoising Higher-order Moments for Blind Digital Modulation Identification in Multiple-antenna Systems}

%
%
%

\author{Sofiane~Kharbech,~\IEEEmembership{Member,~IEEE,}
        Eric~Pierre~Simon,
        Akram~Belazi,
        and~Wei~Xiang,~\IEEEmembership{Senior~Member,~IEEE}

\thanks{The source code for this work is available on \href{https://github.com/sofiane-kharbech/Denoising-HOM-for-DMI}{\emph{https://github.com/sofiane-kharbech/Denoising-HOM-for-DMI}}}

\thanks{S. Kharbech is with the Laboratory IEMN/IRCICA (UMR-CNRS-8520), University of Lille, Lille 59100, France, and with the Laboratory Sys'Com-ENIT (LR-99-ES21), Tunis El Manar University, Tunis 1002, Tunisia (e-mail: sofiane.kharbech@ieee.org).}

\thanks{E. P. Simon is with the Laboratory IEMN/TELICE (UMR-CNRS-8520),
University of Lille, Lille 59100, France (e-mail: eric.simon@univ-lille.fr).}

\thanks{A. Belazi is with the Laboratory RISC (LR-16-ES07), Tunis El Manar University, Tunis 1002, Tunisia (e-mail: akram.belazi@enit.utm.tn).}

\thanks{W. Xiang is with the College of Science and Engineering, James Cook University, Cairns, QLD 4870, Australia (e-mail: wei.xiang@jcu.edu.au).}

}

\maketitle
\thispagestyle{fancy}
\pagestyle{fancy}

\begin{abstract}
The paper proposes a new technique that substantially improves blind digital modulation identification (DMI) algorithms that are based on higher-order statistics (HOS). The proposed technique takes advantage of noise power estimation to make an offset on higher-order moments (HOM), thus getting an estimate of noise-free HOM. When tested for multiple-antenna systems, the proposed method outperforms other DMI algorithms, in terms of identification accuracy, that are based only on cumulants or do not consider HOM denoising, even for a receiver with impairments. The improvement is achieved with the same order of complexity of the common HOS-based DMI algorithms in the same context.
\end{abstract}

\begin{IEEEkeywords}
Cognitive radio, modulation identification, higher-order statistics, multiple-antenna systems, denoising features.
\end{IEEEkeywords}

%
\IEEEpeerreviewmaketitle

\DeclareRobustCommand{\lc}{\raisebox{2pt}{\tikz{\draw[black,solid,line width=1.1pt](0,0) -- (7mm,0);}}}
\DeclareRobustCommand{\lt}{\raisebox{2pt}{\tikz{\draw[black,dashed,line width=1.1pt](0,0) -- (7mm,0);}}}
\DeclareRobustCommand{\lp}{\raisebox{2pt}{\tikz{\draw[black,dotted,line width=1.1pt](0,0) -- (7mm,0);}}}
\DeclareRobustCommand{\lpt}{\raisebox{2pt}{\tikz{\draw[black,dash pattern={on 7pt off 2pt on 1pt off 3pt},line width=1.1pt](0,0) -- (7mm,0);}}}

\section{Introduction}
%
%
%
%
\IEEEPARstart{W}{ith} the continuous and fast development of intelligent communication systems, signal detection is always a critical issue to consider. In intelligent transmission such as cognitive radios, signal detection is no more limited to detecting energy, and it goes beyond, e.g., demodulating unknown signals. Modulation identification is the step that succeeds energy detection and precedes signal demodulation. When both source signals and channel parameters are unknown, we are in a blind context that naturally requires a blind process of modulation recognition. Despite their high identification accuracy, maximum-likelihood-based techniques for modulation identification often suffer from the substantially high complexity. Feature-based algorithms of modulation identification give an alternative that provides a good performance and complexity trade-off.

As low computational complexity features and widely employed in digital modulation identification (DMI), higher-order statistics (HOS), i.e., higher-order moments (HOM) and higher-order cumulants (HOC), have always exhibited a good identification performance \cite{swami2000,kharbech2014,liu2017,hassan2012,bahloul2017,kharbech2013,ali2017,tayakout2018}. Employed HOS in that context are estimated from noisy observations. Estimated HOC are insensitive to noise \cite{mendel1991}, unlike the estimated HOM. As such, most of HOS-based DMI algorithms rely on HOC as features \cite{swami2000,ali2017,tayakout2018,bahloul2017}. However, many other HOS-based DMI algorithms attempt to improve the identification performance by including a set of HOM \cite{kharbech2014,hassan2012,liu2017,kharbech2013}.

Since blind estimation of the noise power is widely addressed in the literature, this motivated us to consider denoising the estimated HOM. The main contribution of this paper is to further improve the performance of a DMI system through the use of noise-free HOM as part of HOS. Furthermore, to enhance the blindness aspect, we make use of a classifier that does not require prior training.
Also, within the framework of this paper, we consider multiple-input–multiple-output (MIMO) systems as an essential part of state-of-the-art wireless systems. Moreover, multi-antenna systems are amply involved in the subject of DMI \cite{kharbech2014,liu2017,hassan2012,bahloul2017,kharbech2013,tayakout2018,marey2014,marey2015,
muhlhaus2013,eldemerdash2016}. As far as we know, there is no yet attempt on offsetting noise in HOM in the blind DMI context. In more detail, the contributions of the paper are three-fold: (i) The derivation of the noise-free HOM formulas for the baseband digitally-modulated signals that in turn require the derivation of the HOM formula for the complex-valued Gaussian noise;
(ii) The denoising approach is integrated into the MIMO system. It allows noiseless HOM for each receive antenna, taking into account the effect of source separation processing;
and (iii) The denoising technique improves the identification accuracy under the influence of various receiver impairments while maintaining the same complexity order of the DMI system.

The rest of the paper is organized as follows. Section II describes the signal model as well as the identification process. In Section III, we give analytical formulas for the denoised moments. Section IV includes a discussion of the presented results. Finally, Section V concludes the paper.

\section{System Model}

In this section we formulate the mathematical model of the received signals and we present a description of each block of the identification process. For better readability, Table {\ref{notation}} defines the notation used in the paper.

\begin{table}[!htb]
\centering
\caption{Notation}
\label{notation}
\begin{tabular}{p{1.2cm} p{6.6cm}}
\hline
$\stackrel{\text{d}}{=}$             & Equality in distribution \\
$\mathbb{C}$      					 & Set of complex numbers \\
$\mathbb{N}^+$						 & Set of natural numbers, 0 is excluded \\
$\jmath$          					 & Imaginary unit \\
$a$, $\textbf{a}$, $\textbf{A}$		 & Scalar, vector, matrix \\
$\textbf{I}_d$    					 & Identity matrix of size $d\times d$ \\
$(.)^T$ 							 & Transpose operator \\
$(.)^*$ 							 & Complex conjugate \\
$(.)^H$								 & Hermitian transpose \\
$|.|$             					 & Modulus of a complex number \\
$\angle{.}$             			 & Argument of a complex number \\
$\hat x$          					 & Estimator of $x$ \\
$\mathcal{N}$						 & Real-valued normal (or Gaussian) distribution \\
$\mathcal{CN}$						 & Complex-valued normal distribution \\
$\mathcal{U}$						 & Uniform distribution \\
$\mathbb{E}\{.\}$					 & Expected value of a random variable \\
$!!$								 & Double factorial \\
$\text{diag}^{-1}(.)$				 & Main diagonal vector of a matrix \\
i.i.d.					 			 & independent and identically distributed \\
\hline
\end{tabular}
\end{table}

\subsection{Signal Model}

We consider a frequency-flat block-fading MIMO system with $N_t$ transmit and $N_r$ receive antennas ($N_t<N_r$). The $n$th received baseband signal at antenna $i$ is expressed as
\begin{equation}
y_i(n)=\sum_{j=1}^{N_t}{h_{ij}x_j(n)}+w_i(n),
\end{equation}
where $y_i(n)$ is the $i$th element of the received MIMO symbol $\textbf{y}(n)\in \mathbb{C}^{N_r\times 1}$, $x_i(n)$ is the $j$th element of the transmitted MIMO symbol $\textbf{x}(n)\in \mathbb{C}^{N_t\times 1}$ (source signals are i.i.d.), $h_{ij}$ is the element $(i,j)$ of the spatially-uncorrelated MIMO channel matrix $\textbf{H}\in \mathbb{C}^{N_r\times N_t}$, and $w_i\sim\mathcal{CN}\left(0,\sigma_w^2\right)$ is a circularly symmetrical complex Gaussian noise at the receive antenna $i$.

\subsection{Identification Process}

The overall process relies on a process that is widely used for the blind DMI issue in MIMO systems \cite{hassan2012,kharbech2013,bahloul2017,liu2017,kharbech2014} while taking advantages of noise power estimation (e.g., \cite{kharbech2014,bahloul2017}) for denoising moments. This is in addition to the use of a blind classifier instead of a classifier that requires prior training.
Fig.~\ref{fig1} depicts the implemented detection process on one receive antenna. The proposed scheme is composed of three main stages, namely: (i) the blind source separation (BSS) step to blindly recover the source signal in conjunction with a noise power estimator; (ii) the denoising-based feature extraction stage allows a better characterization of the modulation scheme (features are denoted by $\mu$ for HOM and by $\kappa$ for HOC); and (iii) the modulation scheme at each antenna is estimated via minimum distance (MD) classification. The estimated modulations $\hat{\theta}_j$ are gathered for the final decision. The most frequent modulation scheme is regarded as the final decision.

For BSS, we make use of the simplified constant modulus algorithm (SCMA) \cite{ikhlef2007}, which is a simplified version of the well-known constant modulus algorithm. SCMA aims at finding a matrix $\textbf{G}\in\mathbb{C}^{N_r\times N_t}$ termed the separator so that the recovered MIMO symbol $\hat{\textbf{x}}(n)$ is estimated as
\begin{equation}
\hat{\textbf{x}}(n)=\textbf{G}^T\textbf{y}(n)=\textbf{G}^T\textbf{H}\textbf{x}(n)+\tilde{\textbf{w}}(n),
\end{equation}
where $\tilde{\textbf{w}}(n)=\textbf{G}^T\textbf{w}(n)$ is the filtered noise. Assuming perfect BSS, i.e., $\textbf{G}^T\textbf{H}=\textbf{I}_{N_t}$, we have
\begin{equation}
\hat{\textbf{x}}(n)=\textbf{x}(n)+\tilde{\textbf{w}}(n).
\end{equation}
The common eigenvalue-based technique \cite{chen1991} is an obvious choice to estimate the noise power and the number of transmitters at once. This technique can be applied as it is in a blind context \cite{kharbech2014}. The features extraction process with denoising, as the main contribution of the paper, is discussed in detail in the next section. For ensuring blind DMI, we make use of a blind classifier, i.e., a classifier that does not need to be trained on test signals with known modulation schemes and particular values of the signal-to-noise ratio (SNR). The MD classifier is the simplest for that goal as it calculates the Euclidian distance of a feature vector with all the theoretical ones, and then selects the closest.

\begin{figure}[!htb]
\centering
\includegraphics[scale=.52]{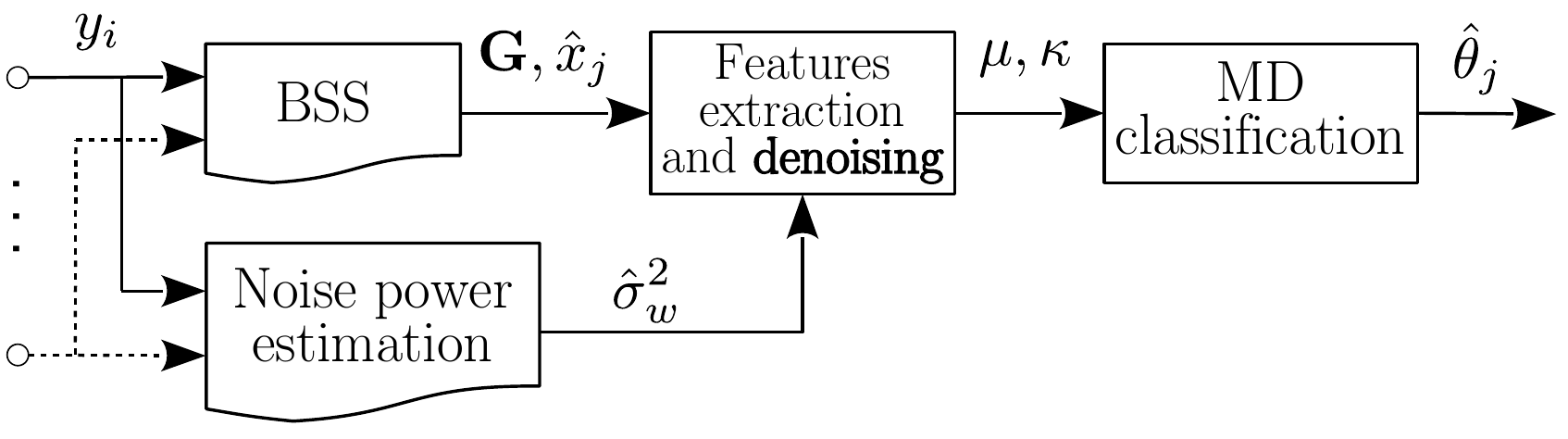}
\caption{Denoising-based blind modulation identification scheme for a receive antenna.}
\label{fig1}
\end{figure}

\section{Higher-order Statistics for DMI and Denoising Moments}

Table \ref{tableHOS} shows how a digital modulation scheme can be characterized by a set of HOS.
Since it is already proved that HOC are noise-insensitive \cite{mendel1991}, we focus on the derivation of the HOM in this section. For a given signal $x$ and integers $p$ and $q$, $0 \leq q \leq p$, the HOM of order $p$ is expressed as
\begin{equation}
\mu_{pq}(x)=\mathbb{E}\left\{x^{p-q}x^{*q}\right\}.
\end{equation}

\begin{table}[!htb]
\caption{Theoretical values of the deployed set of higher-order statistics for the simulated pool of modulation schemes \cite{swami2000,dobre2007}.
These values are obtained using noiseless signals of zero mean and unit variance.}
\label{tableHOS}
\centering
{\small
\begin{tabular}{p{3mm}|c|c|c|c|c|c}
\hline
& \begin{sideways} \mdseries B-PSK \end{sideways} & \begin{sideways} \mdseries Q-PSK \end{sideways} & \begin{sideways} \mdseries 8-PSK \end{sideways} & \begin{sideways} \mdseries 4-ASK \end{sideways} & \begin{sideways} \mdseries 8-ASK \end{sideways} & \begin{sideways} \mdseries 16-QAM \end{sideways} \\
\hline\hline
$\mu_{40}$ & \tablenum{1} & \tablenum{1} & \tablenum{0} & \tablenum{1.64} & \tablenum{1.77} & \tablenum{-0.67} \\
$\mu_{41}$ & \tablenum{1} & \tablenum{0} & \tablenum{0} & \tablenum{1.64} & \tablenum{1.77} & \tablenum{0} \\
$\mu_{42}$ & \tablenum{1} & \tablenum{1} & \tablenum{1} & \tablenum{1.64} & \tablenum{1.77} & \tablenum{1.32} \\
$\mu_{60}$ & \tablenum{1} & \tablenum{0} & \tablenum{0} & \tablenum{2.92} & \tablenum{3.62} & \tablenum{0} \\
$\mu_{61}$ & \tablenum{1} & \tablenum{-1} & \tablenum{0} & \tablenum{2.92} & \tablenum{3.62} & \tablenum{-1.32} \\
$\mu_{62}$ & \tablenum{1} & \tablenum{0} & \tablenum{0} & \tablenum{2.92} & \tablenum{3.62} & \tablenum{0} \\
$\mu_{63}$ & \tablenum{1} & \tablenum{1} & \tablenum{1} & \tablenum{2.92} & \tablenum{3.62} & \tablenum{1.96} \\
$\mu_{84}$ & \tablenum{1} & \tablenum{1} & \tablenum{1} & \tablenum{5.25} & \tablenum{7.92} & \tablenum{3.12} \\
\hline
$\kappa_{40}$ & \tablenum{-2} & \tablenum{1} & \tablenum{0} & \tablenum{-1.36} & \tablenum{-1.24} & \tablenum{-0.68} \\
$\kappa_{41}$ & \tablenum{-2} & \tablenum{0} & \tablenum{0} & \tablenum{-1.36} & \tablenum{-1.24} & \tablenum{0} \\
$\kappa_{42}$ & \tablenum{-2} & \tablenum{-1} & \tablenum{-1} & \tablenum{-1.36} & \tablenum{-1.24} & \tablenum{-0.68} \\
$\kappa_{60}$ & \tablenum{16} & \tablenum{0} & \tablenum{0} & \tablenum{8.32} & \tablenum{7.19} & \tablenum{0} \\
$\kappa_{61}$ & \tablenum{16} & \tablenum{-4} & \tablenum{0} & \tablenum{8.32} & \tablenum{7.19} & \tablenum{2.08} \\
$\kappa_{62}$ & \tablenum{16} & \tablenum{0} & \tablenum{0} & \tablenum{8.32} & \tablenum{7.19} & \tablenum{0} \\
$\kappa_{63}$ & \tablenum{16} & \tablenum{4} & \tablenum{4} & \tablenum{8.32} & \tablenum{7.19} & \tablenum{2.08} \\
\hline
\end{tabular}
}
\end{table}

In fact, not all moments need to be denoised. One can prove that $\mu_{p0}$ is noise-insensitive (cf. Appendix). To derive the noise-free moments for $q>0$, we have to derive their formulas in terms of the noise power $\sigma_w^2$. Towards this end, we will consider $y=x+w$ as a given mix of a digital modulated signal $x$ in the baseband and a circularly symmetrical complex Gaussian noise $w\sim\mathcal{CN}\left(0,\sigma_w^2\right)$ independent of $x$. The considered moments (Table {\ref{tableHOS}}, $q>0$) are derived as follows.

\begin{align}
\begin{split}
\mu_{41}(y)&=\mu_{41}(x)+3\mu_{20}(x)\sigma_w^2
\end{split}\\
\begin{split}
\label{m42}
\mu_{42}(y)&=\mu_{42}(x)+4\mu_{21}(x)\sigma_w^2+\mathbb{E}\left\{w^2w^{*2}\right\}
\end{split}\\
\begin{split}
\label{m61}
\mu_{61}(y)&=\mu_{61}(x)+5\mu_{40}(x)\sigma_w^2
\end{split}\\
\begin{split}
\label{m62}
\mu_{62}(y)&=\mu_{62}(x)+8\mu_{41}(x)\sigma_w^2+6\mu_{20}(x)\mathbb{E}\left\{w^2w^{*2}\right\}
\end{split}\\
\begin{split}
\label{m63}
\mu_{63}(y)&=\mu_{63}(x)+9\mu_{42}(x)\sigma_w^2+9\mu_{21}(x)\mathbb{E}\left\{w^2w^{*2}\right\}\\
&\hspace{3.5mm}+\mathbb{E}\left\{w^3w^{*3}\right\}
\end{split}\\
\begin{split}
\label{m84}
\mu_{84}(y)&=\mu_{84}(x)+16\mu_{63}(x)\sigma_w^2+36\mu_{42}(x)\mathbb{E}\left\{w^2w^{*2}\right\}\\
&\hspace{3.5mm}+16\mu_{21}(x)\mathbb{E}\left\{w^3w^{*3}\right\}+\mathbb{E}\left\{w^4w^{*4}\right\}
\end{split}
\end{align}
In (\ref{m42}), (\ref{m62})--(\ref{m84}), to have formulas in relation to $\sigma_w^2$, we should derive $\mu_{pq}(w)$ (i.e., $\mathbb{E}\left\{w^qw^{*q}\right\}$) for $q>1$. For that purpose, we introduce Theorem \ref{th1}.


\begin{theorem}\label{th1}
Let $s\sim\mathcal{CN}\left(0,\sigma_s^2\right)$,
\begin{equation}
\begin{aligned}
\mu_{pq}(s)=
\left\{
    \begin{array}{rl}
        \left(\dfrac{p}{2}\right)!\sigma_s^p, &\text{if } q=\cfrac{p}{2},\\
        0, &\text{elsewhere}.
    \end{array}
\right.
\end{aligned}
\end{equation}
\end{theorem}

\begin{proof}
$s=s_r+\jmath s_i$, where $s_r$ and $s_i$ are two independent, real-valued normal random variables, i.e., $s_r\stackrel{\text{d}}{=}s_i\sim\mathcal{N}\left(0,\frac{\sigma_s^2}{2}\right)$. For $q=p/2$, we have
\begin{equation}
\begin{aligned}
\mu_{pq}(s)&=\mathbb{E}\left\{|s|^p\right\} &&\text{(cf. Appendix)}\\
&=\mathbb{E}\left\{\left(s_r^2+s_i^2\right)^{p/2}\right\}\\
&=\sum_{k=0}^{p/2}{\dbinom{p/2}{k}\mathbb{E}\left\{s_r^{p-2k}s_i^{2k}\right\}}\\
&=\sum_{k=0}^{p/2}{\dbinom{p/2}{k}\mathbb{E}\left\{s_r^{p-2k}\right\}\mathbb{E}\left\{s_i^{2k}\right\}}\nonumber.
\end{aligned}
\end{equation}
Using the moments of a real-valued normal variable derived in \cite{papoulis1991}, we arrive at
\begin{equation}
\begin{aligned}
\mathbb{E}\left\{s_r^n\right\}=
\left\{
    \begin{array}{rl}
        (n-1)!!\left(\dfrac{\sigma_s^2}{2}\right)^{n/2}, &\text{if $n$ is even},\\
        0, &\text{if $n$ is odd}\nonumber.
    \end{array}
\right.
\end{aligned}
\end{equation}
This leads to
\begin{equation}
\begin{aligned}
\mu_{pq}(s)&=\sum_{k=0}^{p/2}{\dbinom{p/2}{k}\left(p-2k-1\right)!!\left(\frac{\sigma_s^2}{2}\right)^{p/2-k}}  \\
&\hspace{40mm}\times\left(2k-1\right)!!\left(\frac{\sigma_s^2}{2}\right)^k  \\
&=\left(\frac{p}{2}\right)!\left(\frac{\sigma_s^2}{2}\right)^{p/2}\underbrace{\sum_{k=0}^{p/2}{\frac{\left(p-2k-1\right)!!\left(2k-1\right)!!}{k!\left(p/2-k\right)!}}}_{\displaystyle =2^{p/2}}   \\
&=\left(\frac{p}{2}\right)!\sigma_s^p\nonumber.
\end{aligned}
\end{equation}

\end{proof}


Hence, Theorem \ref{th1} results in $\mu_{42}(w)=2\sigma_w^4$, $\mu_{63}(w)=6\sigma_w^6$, and $\mu_{84}(w)=24\sigma_w^8$, and the moments are properly derived in terms of the noise power. Considering our context, the denoised moments at each receive antenna $j$ are given below.
\begin{equation}\label{denmimo}
\begin{aligned}
\mu_{41}(x_j)&=\mu_{41}(\hat{x}_j)-3\mu_{20}(\hat{x}_j)\hat{\sigma}^2_{\widetilde{w}j}\\
\mu_{42}(x_j)&=\mu_{42}(\hat{x}_j)-4\mu_{21}(x_j)\hat{\sigma}^2_{\widetilde{w}j}-2\hat{\sigma}^4_{\widetilde{w}j}\\
\mu_{61}(x_j)&=\mu_{61}(\hat{x}_j)-5\mu_{40}(\hat{x}_j)\hat{\sigma}^2_{\widetilde{w}j}\\
\mu_{62}(x_j)&=\mu_{62}(\hat{x}_j)-8\mu_{41}(x_j)\hat{\sigma}^2_{\widetilde{w}j}-12\mu_{20}(\hat{x}_j)\hat{\sigma}^4_{\widetilde{w}j}\\
\mu_{63}(x_j)&=\mu_{63}(\hat{x}_j)-9\mu_{42}(x_j)\hat{\sigma}^2_{\widetilde{w}j}-18\mu_{21}(x_j)\hat{\sigma}^4_{\widetilde{w}j}-6\hat{\sigma}^6_{\widetilde{w}j}\\
\mu_{84}(x_j)&=\mu_{84}(\hat{x}_j)-16\mu_{63}(x_j)\hat{\sigma}^2_{\widetilde{w}j}\\
&\hspace{3.5mm}-72\mu_{42}(x_j)\hat{\sigma}^4_{\widetilde{w}j}-96\mu_{21}(x_j)\hat{\sigma}^6_{\widetilde{w}j}-24\hat{\sigma}^8_{\widetilde{w}j},
\end{aligned}
\end{equation}
where $\mu_{21}(x_j)=\mu_{21}(\hat{x}_j)-\hat{\sigma}^2_{\widetilde{w}j}$ and $\hat{\sigma}^2_{\widetilde{w}j}$ is the estimated power of the filtered noise at the receive antenna $j$. The variance is estimated as
\begin{equation}
\begin{aligned}
\left[\hat{\sigma}^2_{\widetilde{w}1}, \cdots ,\hat{\sigma}^2_{\widetilde{w}N_t}\right]^T
&=\text{diag}^{-1}\left(\mathbb{E}\left\{\widetilde{\textbf{w}}\widetilde{\textbf{w}}^H\right\}\right) \\
&=\text{diag}^{-1}\left(\mathbb{E}\left\{\textbf{G}^T\textbf{w}\textbf{w}^H\textbf{G}^*\right\}\right)\\
&=\text{diag}^{-1}\left(\textbf{G}^T\mathbb{E}\left\{\textbf{w}\textbf{w}^H\right\}\textbf{G}^*\right)\\
&=\hat{\sigma}_w^2\text{diag}^{-1}\left(\textbf{G}^T\textbf{G}^*\right),
\end{aligned}
\end{equation}
where $\hat{\sigma}_w^2$ is the estimated noise power of the channel. Furthermore, it is worth noting that, to offset the scale factor that can be introduced by BSS non-ideality, all the employed HOS are self-normalized, i.e., divided by $\mu_{21}^{p/2}$.

\section{Numerical Results}

In this section, we evaluate the performance of our proposed scheme as characterized by the probability of correct identification, $P_{ci}$. Computer simulations are based on the modulation pool of Table \ref{tableHOS} and different MIMO antenna configurations.

\begin{figure}[!htb]
\centering
\includegraphics[scale=.65]{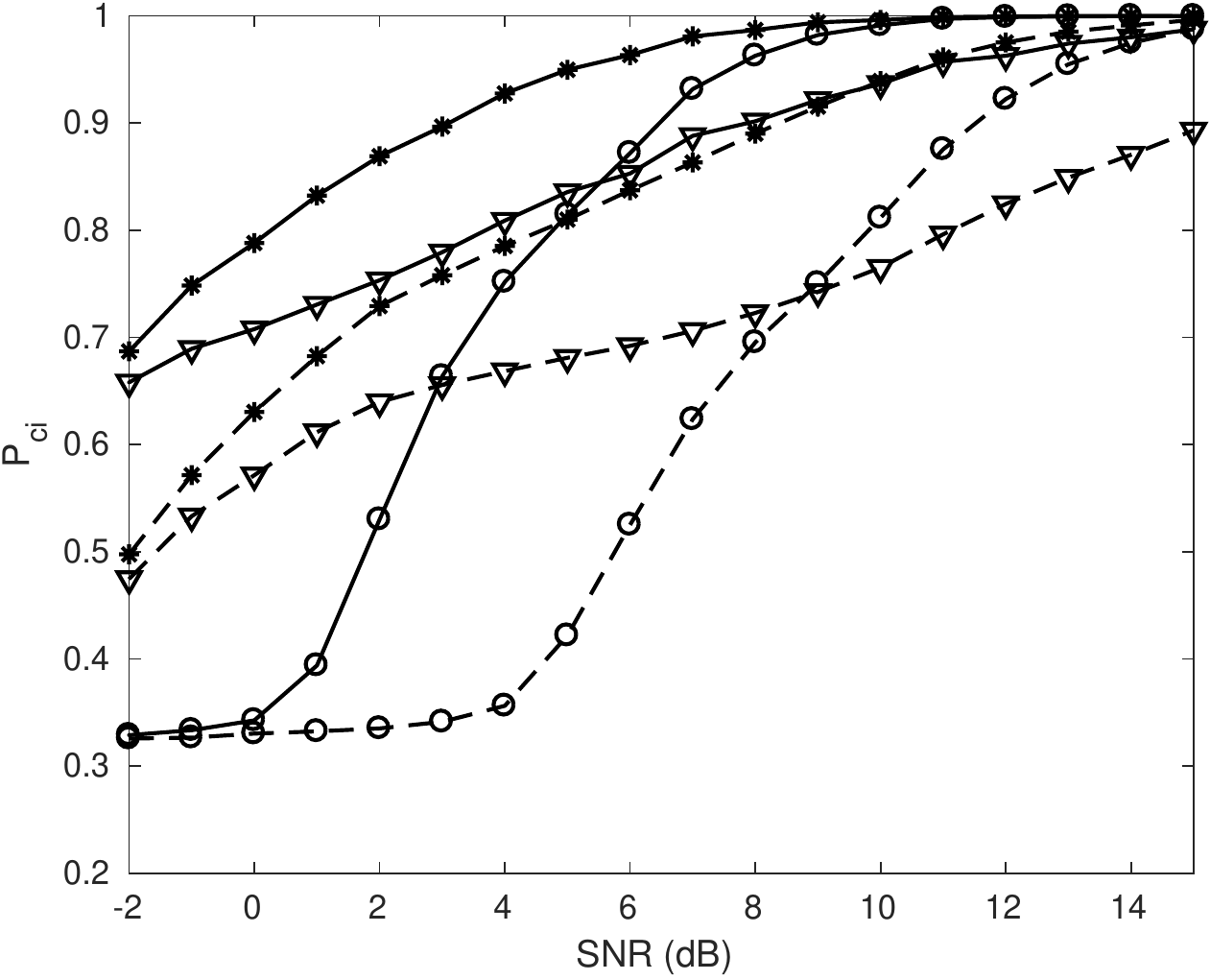}
\caption{Probability of correct identification in terms of SNR. The simulated scenarios are as follows. With moments denoising [$\ast$], without moments denoising [$\circ$], and cumulants only [$\triangledown$]. MIMO antenna configurations are $2\times6$ [\lc] and $3\times6$ [\lt].}
\label{fig2}
\end{figure}

\begin{figure}[!htb]
\centering
\includegraphics[scale=.65]{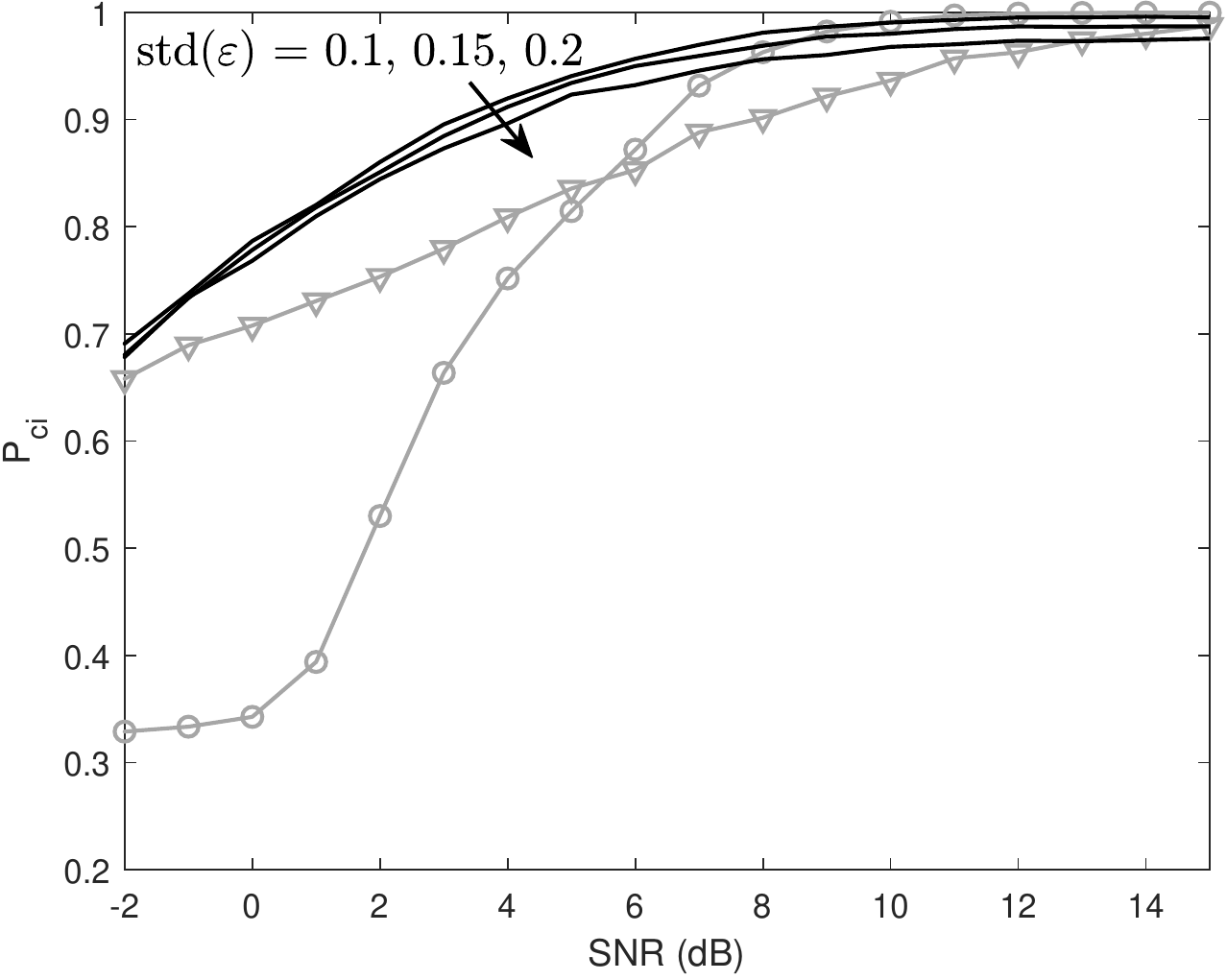}
\caption{Probability of correct identification in terms of SNR under different standard deviations (std) of $\varepsilon$. MIMO antenna configuration is $2\times6$. The plots in gray are copied from Fig.~\ref{fig2} for comparison with the HOC-only and non-denoised HOM senarios.}
\label{fig3}
\end{figure}

\begin{figure*}[!htb]
\centering
\subfloat[Phase noise effect.]{\includegraphics[scale=0.65]{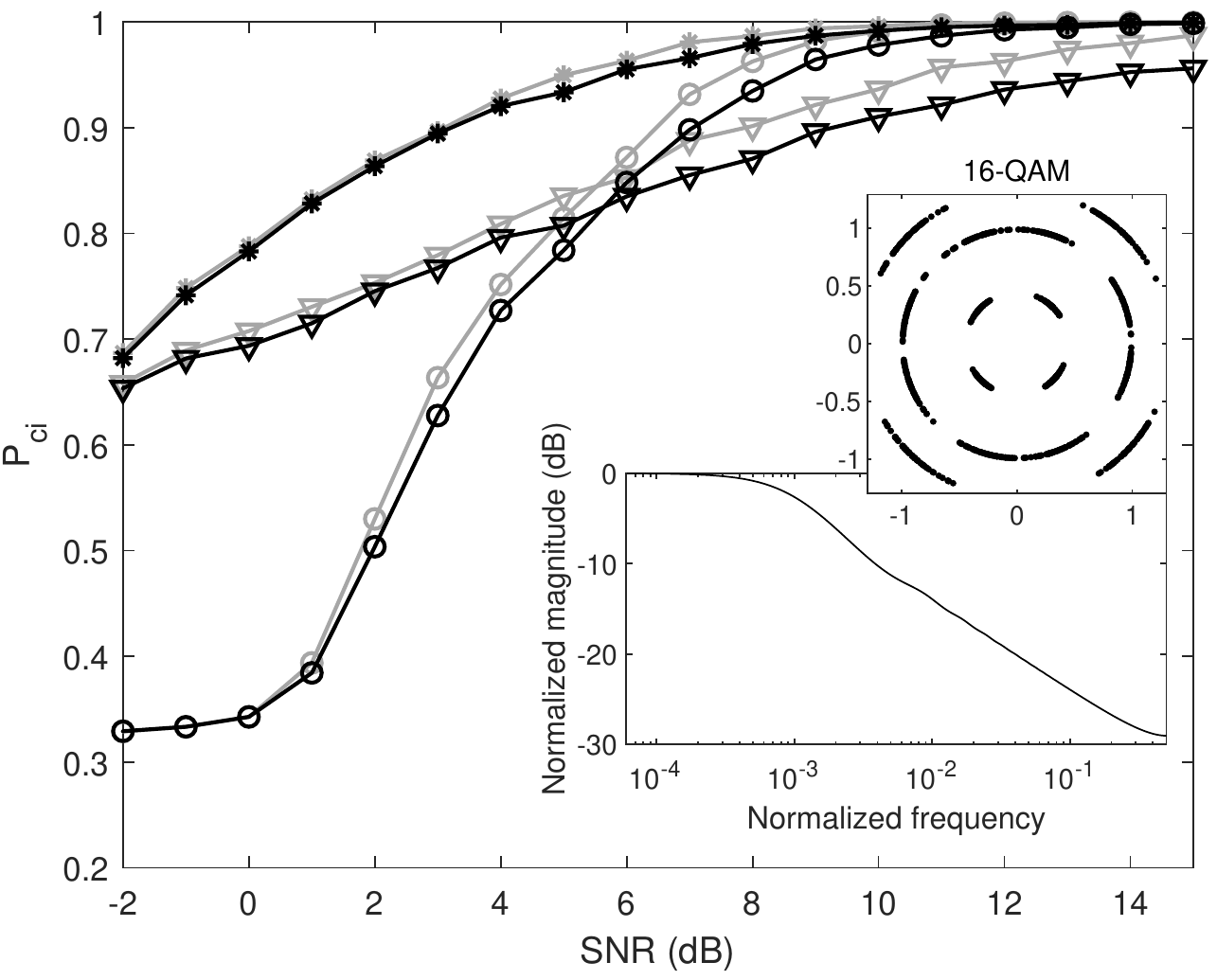} \quad}
\subfloat[CFO effect with normalized frequency offset of the order $10^{-4}$.]{\includegraphics[scale=0.65]{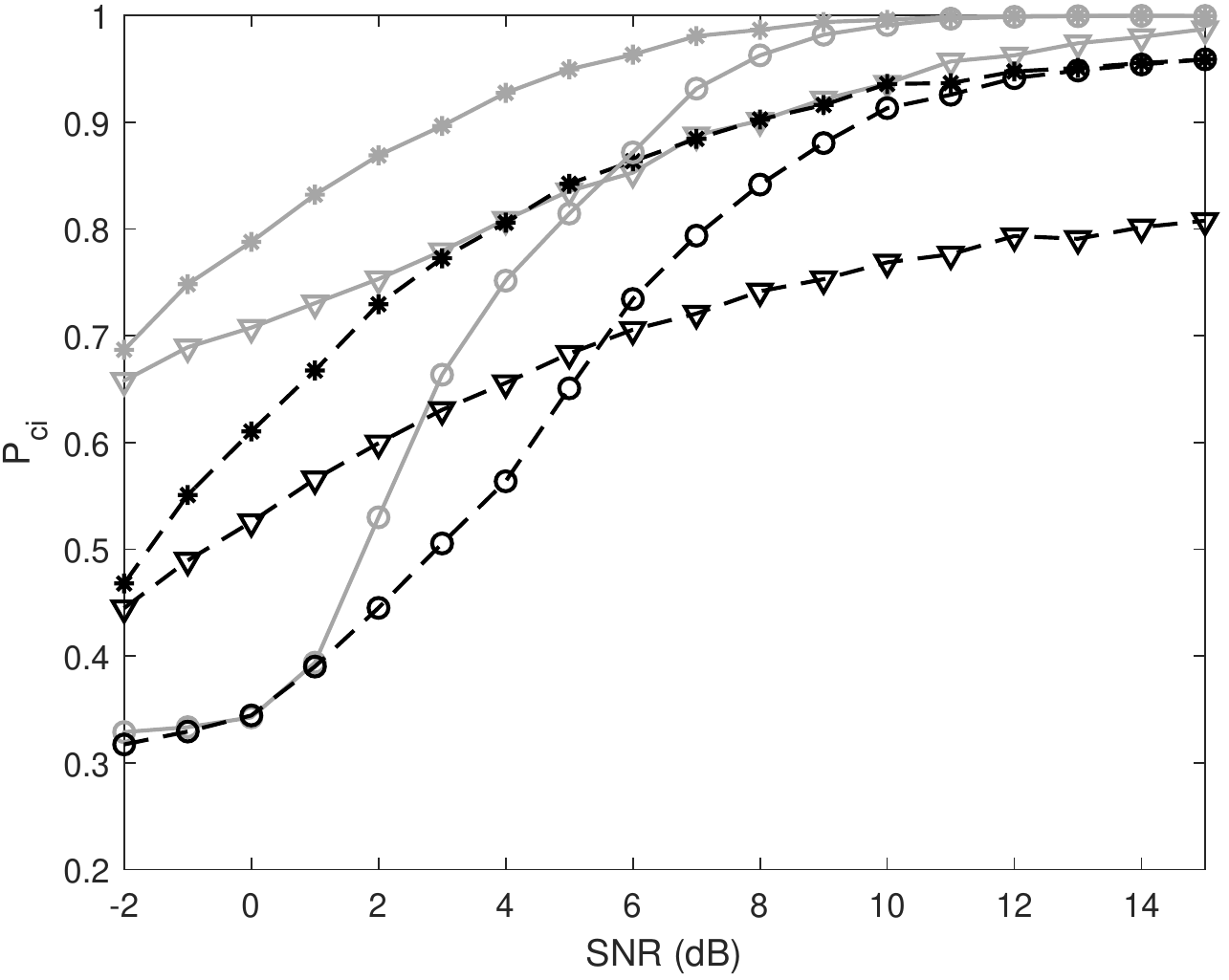}}
\caption{Probability of correct identification in terms of SNR in consideration of local oscillator imperfections. The simulated scenarios include: with moments denoising [$\ast$], without moments denoising [$\circ$], and cumulants only [$\triangledown$]. The MIMO antenna configuration is $2\times6$. The plots in gray are copied from Fig.~\ref{fig2} for comparison. The plots inside (a) show the frequency response of the filter used to generate the phase noise (characteristics of the power spectrum density mask are $2.10^{-3}$ and $-3$ dBc/Hz for the phase noise normalized-bandwidth and level, respectively), and the related effect observed on noise-free 16-QAM constellation as an example.}
\label{fig4}
\end{figure*}

Fig.~\ref{fig2} shows the performance of the identification system, with perfect estimation of $\sigma_w^2$ for the following scenarios: (i) a set of HOS is used without denoising HOM; (ii) only a set of HOC is used; and (iii) a set of HOS is used with denoising HOM. Distinctly, the third scenario, which represents our proposal, exhibits a better performance than the other comparative ones. For example, considering the MIMO antenna configuration $2\times6$ within the simulated SNR range, DMI in scenario (iii) attains an average gain of about 15\% and 7\% compared to a DMI in scenarios (i) and (ii), respectively. Regarding scenarios (i) and (iii), the performance gain is further higher at lower SNRs. It is noted in the same figure that for all the investigated scenarios, the identification performance decreases when $\Delta=N_r-N_t$ decreases. In fact, this phenomenon is due to the proportionality between the effectiveness of the BSS and $\Delta$ \cite{kharbech2014}. In respect of the performance drop due to a lower $\Delta$, DMI endures an average performance loss of 10\% for scenario (iii) compared to about 17\% and 15\% for scenarios (i) and (ii), respectively. This means that the performance gain in scenario (iii) is more significant for the MIMO antenna configuration $3\times6$. Therefore, the proposed method is more resistant against the BSS impairments when $\Delta$ becomes smaller.

Robustness to imperfect estimation of the channel noise power (i.e., $\hat{\sigma}_w^2=\sigma_w^2+\varepsilon$, where $\varepsilon$ is the estimation error) is investigated in Fig.~\ref{fig3}. Obviously, the performance undergoes degradation as the variance of the channel's noise power estimator becomes higher. Nevertheless, the denoising remains relevant, in particular at lower SNRs, and still outperforms the HOC-only approach.

Fig.~\ref{fig4} assess the reliability of the denoising-based process against impairments of the baseband receiver (local oscillator). To generate the phase noise, we employ one of the most commonly used procedures \cite{kasdin1995}.
Globally, due to the robustness of the utilized BSS algorithm over the phase noise, the three scenarios have not endured a significant performance loss. However, the denoising-based approach has almost maintained the same performance. This is unlike the effect of the carrier frequency offset (CFO), where a drop in performance is clearly observed for all approaches. Indeed, this is due to deficient BSS since the CFO induces a time-variation effect, and the SCMA algorithm is designed for time-invariant channels \cite{kharbech2016}. Yet, the DMI algorithm with the denoising approach remains the best compared to the other ones.

\section{Conclusion}

In summary, after deriving the noise-free moments for DMI, simulation results proved that modulation detection through denoising is more efficient that the deployment of classical processes for detection like the HOC-based ones and other schemes that do not proceed to denoise moments. Moreover, as it did not assume particular values of the SNR, the proposed approach of denoising is also applicable to hierarchical classification. With regards to complexity, additional operations related to denoising terms have a constant time complexity. Besides, since the evaluated identification scheme is based on the ones like \cite{kharbech2014,hassan2012} which follow the same complexity order (both have a polynomial running time), the outstanding performance of our proposal is achieved more blindly and with the same order of computational complexity of the HOS-based DMI systems in MIMO channels. More generally, employing denoised HOM for any HOS-based DMI algorithm, which assumes \emph{a priori} knowledge of the noise power or considers its estimation (e.g., \cite{kharbech2014,hassan2012,muhlhaus2013}), improves the identification accuracy without an increase in the computational complexity or at least its order.

\appendix[Deriving $\mu_{p0}(x+w)$ and $\mu_{pq}(s)$]

\paragraph{Deriving $\mu_{p0}(x+w)$}

Let the mix $x+w$ of a baseband digitally-modulated signal $x$ and a noise $w\sim\mathcal{CN}\left(0,\sigma_w^2\right)$ independent of $x$,
\begin{equation}
\begin{aligned}
\mu_{p0}(x+w)&=\mathbb{E}\left\{\left(x+w\right)^p\right\}=\sum_{k=0}^p{\dbinom{p}{k}\mathbb{E}\left\{x^{p-k}w^k\right\}}\\
&=\sum_{k=0}^p{\dbinom{p}{k}\mathbb{E}\left\{x^{p-k}\right\}\mathbb{E}\left\{w^k\right\}}.
\end{aligned}
\end{equation}
However, $w$ is a circularly symmetrical Gaussian process, i.e., $\mathbb{E}\left\{w^k\right\}=0 ~ \forall k\in\mathbb{N}^+$, thus,
\begin{equation}
\begin{aligned}
\mu_{p0}(x+w)&=\mathbb{E}\left\{x^p\right\}=\mu_{p0}(x).
\end{aligned}
\end{equation}

\paragraph{Deriving $\mu_{pq}(s)$}

Let $s=re^{\jmath\theta}\sim\mathcal{CN}\left(0,\sigma_s^2\right)$, where $r=|s|$ and $\theta=\angle{s}$ are independent,
\begin{equation}
\begin{aligned}
\mu_{pq}(s)&=\mathbb{E}\left\{s^{p-q}s^{*q}\right\}\\
&=\mathbb{E}\left\{r^{p-q}e^{\jmath(p-q)\theta}r^qe^{-\jmath q\theta}\right\}\\
&=\mathbb{E}\left\{r^p\right\}\mathbb{E}\left\{e^{\jmath(p-2q)\theta}\right\},
\end{aligned}
\end{equation}
however, $\theta\sim\mathcal{U}\left(-\pi,\pi\right)$, consequently,
\begin{equation}
\begin{aligned}
\mathbb{E}\left\{e^{\jmath(p-2q)\theta}\right\}=
\left\{
    \begin{array}{rl}
        1, &\text{if } q=\cfrac{p}{2},\\
        0, &\text{elsewhere}.
    \end{array}
\right.
\end{aligned}
\end{equation}
Thus,
\begin{equation}
\begin{aligned}
\mu_{pq}(s)=
\left\{
    \begin{array}{rl}
        \mathbb{E}\left\{|s|^p\right\}, &\text{if } q=\cfrac{p}{2},\\
        0, &\text{elsewhere}.
    \end{array}
\right.
\end{aligned}
\end{equation}

\section*{Acknowledgment}

The authors would like to thank Dr. Octavia A. Dobre from Memorial University of Newfoundland (NL, Canada) and Dr. Chad M. Spooner from NorthWest Research Associates (CA, USA) for the useful discussions on higher-order statistics.

\ifCLASSOPTIONcaptionsoff
  \newpage
\fi



%

\bibliographystyle{IEEEtran}
\bibliography{preprint_wcl_2020}

\begin{thebibliography}{10}
\providecommand{\url}[1]{#1}
\csname url@samestyle\endcsname
\providecommand{\newblock}{\relax}
\providecommand{\bibinfo}[2]{#2}
\providecommand{\BIBentrySTDinterwordspacing}{\spaceskip=0pt\relax}
\providecommand{\BIBentryALTinterwordstretchfactor}{4}
\providecommand{\BIBentryALTinterwordspacing}{\spaceskip=\fontdimen2\font plus
\BIBentryALTinterwordstretchfactor\fontdimen3\font minus
  \fontdimen4\font\relax}
\providecommand{\BIBforeignlanguage}[2]{{%
\expandafter\ifx\csname l@#1\endcsname\relax
\typeout{** WARNING: IEEEtran.bst: No hyphenation pattern has been}%
\typeout{** loaded for the language `#1'. Using the pattern for}%
\typeout{** the default language instead.}%
\else
\language=\csname l@#1\endcsname
\fi
#2}}
\providecommand{\BIBdecl}{\relax}
\BIBdecl

\bibitem{swami2000}
A.~{Swami} and B.~M. {Sadler}, ``Hierarchical digital modulation classification
  using cumulants,'' \emph{IEEE Transactions on Communications}, vol.~48,
  no.~3, pp. 416--429, Mar. 2000.

\bibitem{kharbech2014}
S.~{Kharbech}, I.~{Dayoub}, M.~{Zwingelstein-Colin}, E.~P. {Simon}, and
  K.~{Hassan}, ``Blind digital modulation identification for time-selective
  mimo channels,'' \emph{IEEE Wireless Communications Letters}, vol.~3, no.~4,
  pp. 373--376, Aug. 2014.

\bibitem{liu2017}
X.~{Liu}, C.~{Zhao}, P.~{Wang}, Y.~{Zhang}, and T.~{Yang}, ``Blind modulation
  classification algorithm based on machine learning for spatially correlated
  mimo system,'' \emph{IET Communications}, vol.~11, no.~7, pp. 1000--1007, May
  2017.

\bibitem{hassan2012}
K.~{Hassan}, I.~{Dayoub}, W.~{Hamouda}, C.~N. {Nz\'eza}, and M.~{Berbineau},
  ``Blind digital modulation identification for spatially-correlated mimo
  systems,'' \emph{IEEE Transactions on Wireless Communications}, vol.~11,
  no.~2, pp. 683--693, Feb. 2012.

\bibitem{bahloul2017}
M.~R. Bahloul, M.~Z. Yusoff, A.-H. Abdel-Aty, M.~N. Saad, and A.~Laouiti,
  ``Efficient and reliable modulation classification for mimo systems,''
  \emph{Arabian Journal for Science and Engineering}, vol.~42, no.~12, pp.
  5201--5209, Dec. 2017.

\bibitem{kharbech2013}
S.~Kharbech, I.~Dayoub, E.~Simon, and M.~Zwingelstein-Colin, ``Blind digital
  modulation detector for {MIMO} systems over high-speed railway channels,'' in
  \emph{International Workshop on Communication Technologies for
  Vehicles}.\hskip 1em plus 0.5em minus 0.4em\relax Springer, May 2013, pp.
  232--241.

\bibitem{ali2017}
A.~{Ali} and F.~{Yangyu}, ``Automatic modulation classification using deep
  learning based on sparse autoencoders with nonnegativity constraints,''
  \emph{IEEE Signal Processing Letters}, vol.~24, no.~11, pp. 1626--1630, Nov.
  2017.

\bibitem{tayakout2018}
H.~{Tayakout}, I.~{Dayoub}, K.~{Ghanem}, and H.~{Bousbia-Salah}, ``Automatic
  modulation classification for d-stbc cooperative relaying networks,''
  \emph{IEEE Wireless Communications Letters}, vol.~7, no.~5, pp. 780--783,
  Oct. 2018.

\bibitem{mendel1991}
J.~M. {Mendel}, ``Tutorial on higher-order statistics (spectra) in signal
  processing and system theory: theoretical results and some applications,''
  \emph{Proceedings of the IEEE}, vol.~79, no.~3, pp. 278--305, Mar. 1991.

\bibitem{marey2014}
M.~{Marey} and O.~A. {Dobre}, ``Blind modulation classification algorithm for
  single and multiple-antenna systems over frequency-selective channels,''
  \emph{IEEE Signal Processing Letters}, vol.~21, no.~9, pp. 1098--1102, Sep.
  2014.

\bibitem{marey2015}
------, ``Blind modulation classification for alamouti stbc system with
  transmission impairments,'' \emph{IEEE Wireless Communications Letters},
  vol.~4, no.~5, pp. 521--524, Oct. 2015.

\bibitem{muhlhaus2013}
M.~S. {M\"uhlhaus}, M.~{\"Oner}, O.~A. {Dobre}, and F.~K. {Jondral}, ``A low
  complexity modulation classification algorithm for mimo systems,'' \emph{IEEE
  Communications Letters}, vol.~17, no.~10, pp. 1881--1884, Oct. 2013.

\bibitem{eldemerdash2016}
Y.~A. {Eldemerdash}, O.~A. {Dobre}, and M.~{\"Oner}, ``Signal identification
  for multiple-antenna wireless systems: achievements and challenges,''
  \emph{IEEE Communications Surveys Tutorials}, vol.~18, no.~3, pp. 1524--1551,
  thirdquarter 2016.

\bibitem{ikhlef2007}
A.~Ikhlef and D.~Le~Guennec, ``A simplified constant modulus algorithm for
  blind recovery of mimo qam and psk signals: A criterion with convergence
  analysis,'' \emph{EURASIP Journal on Wireless Communications and Networking},
  vol. 2007, no.~1, p. 090401, Dec. 2007.

\bibitem{chen1991}
W.~{Chen}, K.~M. {Wong}, and J.~P. {Reilly}, ``Detection of the number of
  signals: a predicted eigen-threshold approach,'' \emph{IEEE Transactions on
  Signal Processing}, vol.~39, no.~5, pp. 1088--1098, May 1991.

\bibitem{dobre2007}
O.~Dobre, A.~Abdi, Y.~Bar-Ness, and W.~Su, ``Survey of automatic modulation
  classification techniques: classical approaches and new trends,'' \emph{IET
  Communications}, vol.~1, no.~2, pp. 137--156, Apr. 2007.

\bibitem{papoulis1991}
A.~Papoulis, \emph{Probability, Random Variables, and Stochastic Processes},
  3rd~ed.\hskip 1em plus 0.5em minus 0.4em\relax McGraw-Hill, Inc., 1991, ch.
  Functions of One Random Variable, pp. 109--111.

\bibitem{kasdin1995}
N.~J. {Kasdin}, ``Discrete simulation of colored noise and stochastic processes
  and 1/f/sup /spl alpha// power law noise generation,'' \emph{Proceedings of
  the IEEE}, vol.~83, no.~5, pp. 802--827, May 1995.

\bibitem{kharbech2016}
S.~{Kharbech}, I.~{Dayoub}, M.~{Zwingelstein-Colin}, and E.~P. {Simon}, ``On
  classifiers for blind feature-based automatic modulation classification over
  multiple-input-multiple-output channels,'' \emph{IET Communications},
  vol.~10, no.~7, pp. 790--795, Apr. 2016.

\end{thebibliography}

\end{document}